\documentclass[11pt]{article}

\usepackage{enumitem} 
\setitemize{itemsep=-4pt} 
\setenumerate{itemsep=-4pt} 

\usepackage[margin=1cm]{caption} 
\usepackage{subfig}
\usepackage{tikz}
\usetikzlibrary{arrows,backgrounds,calc,fit,shapes.geometric}

\usepackage{mathtools} 

\tikzset{every fit/.append style=text badly centered}

\usepackage{fullpage}
\usepackage{amsmath}
\usepackage{amsfonts}
\usepackage{amssymb}
\usepackage{amsthm}
\usepackage{mathrsfs}

\newtheorem{theorem}{Theorem}[section]

\newtheorem{corollary}[theorem]{Corollary}

\newtheorem{definition}[theorem]{Definition}

\newtheorem{lemma}[theorem]{Lemma}

\usepackage[bookmarks=true,hypertexnames=false,pagebackref]{hyperref}
\hypersetup{colorlinks=true, citecolor=blue, linkcolor=red, urlcolor=blue}

\newcommand{\arity}{\operatorname{arity}}

\def\borderColor{blue!60}

\tikzstyle{internal} = [draw, fill, shape=circle]
\tikzstyle{external} = [shape=circle]
\tikzstyle{square}   = [draw, fill, rectangle]
\tikzstyle{triangle} = [draw, fill, regular polygon, regular polygon sides=3, inner sep=2.5pt] 

\newcommand{\newfontobj}[2]{
  \newcommand{#1}[1]{
    \expandafter\def\csname##1\endcsname{{#2 ##1}}}}

\newcommand{\BPP}{\texttt{BPP}}
\newcommand{\BQP}{\texttt{BQP}}

\newcommand{\numP}{\#\texttt{P}}

\newcommand{\transpose}[1]{#1^\texttt{T}}

\newcommand{\tensor}{\otimes}

\title{Clifford Gates in the Holant Framework}

\author{
 Jin-Yi Cai\\
 \scriptsize University of Wisconsin-Madison\\
 \footnotesize \texttt{jyc@cs.wisc.edu}
 \and
 Heng Guo\\
 \scriptsize Queen Mary, University of London\\
 \footnotesize \texttt{h.guo@qmul.ac.uk}
 \and
 Tyson Williams\\
 \scriptsize  Blocher Consulting\\
 \footnotesize \texttt{tdw@cs.wisc.edu}
}

\begin{document}
\maketitle

\begin{abstract}
  We show that the Clifford gates and stabilizer circuits in the quantum computing literature, 
  which admit efficient classical simulation,
  are equivalent to affine signatures under a unitary condition.
  The latter is a known class of tractable functions under the Holant framework.
\end{abstract}

\section{Introduction}

Ever since Shor's famous quantum algorithm to factor integers in polynomial time~\cite{Sho99},
researchers have sought to understand the relationship between efficient classical computation,
as represented by the class \BPP, and efficient quantum computation, as represented by the class $\BQP$.
Since the quantum computational model is typically presented as a uniform family of quantum circuits, the question becomes,
``which classes of quantum circuits can be simulated by a classical computer in polynomial time and which ones cannot?''
It is widely believed that $\BPP \ne \BQP$, which implies that not all quantum circuits have efficient classical simulations.
The goal is then to determine under what restrictions this becomes possible.

There have been many successes along these lines by considering various bounds on
measures of entanglement~\cite{JL03, Vid03, SDV06}, width~\cite{YS06, MS08}, and planarity~\cite{Joz06, Bra09}.
Another popular restriction is to limit the allowed gates, which is the point of view we take in this paper.
For example, it is easy to see that circuits composed of only 1-qubit gates
can be simulated in deterministic polynomial time by following each qubit independently.
We call such circuits \emph{degenerate}.

Other than this trivial example, there are basically two known classes of gates that compose circuits admitting efficient classical simulations.
The first is Clifford gates, which get their name by a connection with Clifford algebras.
Circuits using only these gates are called \emph{stabilizer} circuits because the Clifford gates stabilize the group of Pauli matrices.
The efficient classical simulation of Clifford gates is known as the Gottesman-Knill Theorem~\cite{Got99},
which now has several alternative proofs~\cite{AG04, AB06, vdN10, JN14}.
The second is \emph{matchgates} with nearest neighbor interactions.
Valiant \cite{Val02} gave an efficient classical simulation of this class by reducing to counting perfect matchings in planar graphs,
which is computable in deterministic polynomial time due to the FKT algorithm~\cite{TF61, Kas67}.
For the special case of 2-qubit matchgates with even support,
many alternative proofs of tractability and characterizations were found~\cite{TD02, Kni01, Joz08, JKMW10, vdN11}.

Valiant further developed the idea of matchgates, resulting in holographic algorithms \cite{Val08},
which give a number of surprising algorithms to problems that were not known to be tractable previously. 
Holographic algorithms have been further developed and generalized extensively \cite{Val06, CL11, CLX13a, CGW14a, Chen16, Xia16}.
To understand the surprising power of holographic algorithms, the Holant framework was proposed by Cai, Lu, and Xia~\cite{CLX11d}.
The exact definition of the Holant framework is beyond this paper, and here we will just give a simple description.
Details can be found in \cite{CLX11d,CGW16}.
Roughly speaking, a Holant problem is to evaluate the contraction of a tensor network,
and its computational complexity is determined by the set of tensors we allow.
Equivalently, it is a constraint satisfaction problem defined on a graph,
where we put (constraint) functions on vertices and edges are variables, 
and the goal is to evaluate the total sum over all possible assignments on edges of the product weights of functions on vertices.
There has been considerable success towards classifying the computational complexity of Holant problems~\cite{CLX11d, CHL12, HL16, CGW16, CFGW15}.
These classification results provide us, under various settings, thorough understandings of what class of problems is polynomial time
tractable and what is \numP-hard.

The original motivation of matchgates~\cite{Val02} was to classically simulate quantum circuits.
Under the Holant framework, the matchgates correspond to a special case of
 the known tractable classes \cite{CFGW15}.
However, it is natural to wonder, whether the other class of
 quantum circuits that admits efficient simulation,
namely Clifford gates / stabilizer circuits, is connected to Holant problems as well.

We show that Clifford gates are indeed also a special case of a known tractable class, called affine signatures \cite{CLX14,CGW16}.
The quantum circuit model requires us to view a function (or tensor) as a unitary gate,
namely that it has the same number of inputs and outputs, and the matrix representation is unitary.
In fact, we show that the Clifford gates are exactly the class of affine signatures under these two restrictions,
namely, with specified inputs and outputs, and a corresponding unitary matrix representation (see Theorem \ref{thm:equiv}).\footnote{It turns out that for affine signatures, unitary matrix requirements are equivalent to requiring simply non-singularity. See Corollary \ref{cor:nonsingluar:unitary}.}
Since affine signatures are explicitly defined, 
our result gives an alternative and explicit characterization of 
Clifford gates. 

Classifying the classical complexity of quantum circuits is a well-known open problem (see Question 2 of \cite{Aar05}).
Valiant \cite{Val02} has shown how to reduce from computing the marginal probability of a certain output for a quantum circuit,
to evaluating a corresponding Holant problem.
Thus, tractable Holant problems can potentially provide new simulation algorithms.
For example, the efficient simulation of Clifford gates can be obtained through the polynomial-time algorithm for affine signatures \cite{CLX14},
which provides yet another proof of the Goettesman-Knill theorem \cite{Got99}.
However, the study of Holant problems has been mostly focused on symmetric functions \cite{CGW16, CFGW15}, and with the unitary matrix restriction, 
known tractable families become exactly those corresponding to degenerate, matchgates, and stabilizer circuits.
Despite some recent progress \cite{LW16, CLX17, Backens17}, 
a general dichotomy theorem for Holant problems remains elusive.
We hope that studies in this direction can shed some light in classifying quantum circuits in the future.

The results reported here were obtained in 2012,
and the materials had been presented at 
the Simons Institute for the Theory of Computing.
Recently Backens \cite{Backens17} made some new connections between quantum computation and Holant problems, 
and used them to obtain a new dichotomy result under the Holant framework.
In light of this new exciting connection, we are encouraged to publish 
these results.

\section{Definitions and the Main Result} \label{sec:affine}

We first introduce several definitions so that we can state our main theorem.
The two basic objects we look at are affine signatures, which are a tractable class of constraint functions in the Holant framework~\cite{CLX14},
and Clifford gates, which compose stabilizer circuits that can be efficiently and classically simulated.

\subsection{Affine Signatures} \label{subsec:affine_def}

\begin{definition} \label{def:affine_sigs}
 A $k$-ary function $f(x_1, \dotsc, x_k)$ is \emph{affine} if it has the form
 \[\lambda \chi_{A \mathbf{x} = 0} \cdot \sqrt{-1}^{\sum_{j=1}^n \langle \boldsymbol{\alpha_j}, \mathbf{x} \rangle}\]
 where $\lambda \in \mathbb{C}$, $\mathbf{x} = \transpose{(x_1, x_2, \dotsc, x_k, 1)}$, $A$ is matrix over $\mathbb{F}_2$,
 $\boldsymbol{\alpha_j}$ is a vector over $\mathbb{F}_2$, and $\chi$ is a 0-1 indicator function such that $\chi_{A \mathbf{x} = 0}$ is 1 iff $A \mathbf{x} = 0$.
 Note that the dot product $\langle \boldsymbol{\alpha_j}, \mathbf{x} \rangle$ is calculated over $\mathbb{F}_2$,
 while the summation $\sum_{j=1}^n$ on the exponent of $i = \sqrt{-1}$ is evaluated as a sum mod 4 of 0-1 terms.
 We use $\mathscr{A}$ to denote the set of all affine functions.
\end{definition}

In other words, the affine signatures are, up to a scalar,
the signatures with affine support whose nonzero entries are expressible as $i$ raised to a sum of linear indicator functions.
This alternative view is the original definition of the affine signatures \cite[Definition 3.1]{CLX14}.
For the indicator function $\chi_{A \mathbf{x} = 0}$, we also write $\chi_{A \mathbf{x}}$ for brevity.

A sum of 0-1 indicator functions defined by the affine linear functions $L_j$ atop $i$ has an alternate form.
We can express $\lambda i^{\sum_{j=1}^n \langle \boldsymbol{\alpha_j}, \mathbf{x} \rangle}$ as $\lambda' i^{Q(\mathbf{x})}$,
where $Q$ is a homogeneous quadratic polynomial over $\mathbb{Z}$
with the additional requirement that every cross term $x_j x_\ell$ with $j \ne \ell$ has an even coefficient.
To see this, we observe that the 0-1 indicator function of an affine linear function $L(\mathbf{x})$ can be replaced by $(L(\mathbf{x}))^2$, 
and $L(\mathbf{x}) = 0,1 \pmod{2}$ if and only if $(L(\mathbf{x}))^2 = 0,1 \pmod{4}$, respectively.
Then every cross term has an even coefficient and $i^{x_j}$ can be replaced by $i^{x_j^2}$.
Conversely, we can express $Q \bmod{4}$ as a sum of squares of linear forms of $\mathbf{x}$ using the extra condition that all cross terms have an even coefficient.
We utilize this definition in the present work.

From the relation $A \mathbf{x} = 0$, where $\mathbf{x} = \transpose{(x_1, x_2, \dotsc, x_k, 1)}$,
we can pick a maximal set of free variables, such that all other variables can be expressed by their affine linear sums mod~2.
Say we have $\mathbf{x''} = A' \mathbf{x'} + \mathbf{b}$ and there are $r$ variables in $\mathbf{x'}$.
This is a relation in $\mathbb{F}_2$.
We can use it to replace variables $\mathbf{x''}$ by $\mathbf{x'}$ in the expression $\lambda' i^{Q(\mathbf{x})}$.
This is valid because for every symmetric integer matrix $S$,
if $\mathbf{x} \equiv \mathbf{y} \pmod{2}$,
then $\transpose{\mathbf{x}} S \mathbf{x} \equiv \transpose{\mathbf{y}} S \mathbf{y} \pmod{4}$.
After replacing $i^{x_j}$ by $i^{x_j^2}$, and absorbing a constant term $i^c$,
we may assume the function is $\mu \chi_{A \mathbf{x} =0} i^{\transpose{\mathbf{x'}} C \mathbf{x'}}$,
where $C$ is an $r$-by-$r$ symmetric matrix, and $\mu$ has the same norm as $\lambda$.


\subsection{Unitary Affine Signatures}

In order to state what makes a signature a valid quantum operation,
we define a mapping from a signature of even arity to the matrix that 
represents it as its (weighted) truth table.

\begin{definition} \label{def:sig_matrix}
 For a signature $f$ of arity $2n$ on Boolean inputs $x_1, \dotsc, x_{2n}$,
 the \emph{signature matrix} of $f$ is the $2^n$-by-$2^n$ matrix $M_f$ where
 the entry in the row indexed by $(x_1, \dotsc, x_n) \in \{0, 1\}^n$ and column indexed by $(x_{2n}, \dotsc, x_{n+1}) \in \{0, 1\}^n$ is $f(x_1, \dotsc, x_{2n})$.
\end{definition}

\begin{figure}[hbtp]
  \begin{minipage}[t]{.45\textwidth}
    \begin{center}
      \begin{tikzpicture}[scale=0.55,transform shape,semithick]
        \draw (0,0) node [] (0) {};
        \draw (0,3) node [] (4) {};
        \draw (-2,0) -- (-0.52,0);
        \draw (-2,2) -- (-0.52,2);
        \draw (-2,3) -- (-0.52,3);
        \draw (-1.5,0.75) node [] {\Huge $\vdots$};
        \draw (-1.5,1.25) node [] {\Huge $\vdots$};
        \draw (2,0) -- (0.52,0);
        \draw (2,2) -- (0.52,2);
        \draw (2,3) -- (0.52,3);
        \draw (1.5,0.75) node [] {\Huge $\vdots$};
        \draw (1.5,1.25) node [] {\Huge $\vdots$};

        \draw (-2.5,3) node [] {\Large $x_1$};
        \draw (-2.5,2) node [] {\Large $x_2$};
        \draw (-2.5,0) node [] {\Large $x_n$};
        
        \draw (2.4,3) node [] {\Large $x_{2n}$};
        \draw (2.6,2) node [] {\Large $x_{2n-1}$};
        \draw (2.6,0) node [] {\Large $x_{n+1}$};

        \draw (0,1.5) node [] {\LARGE $f$};

        \begin{pgfonlayer}{background}
          \node[draw=black,thick,rounded corners,fit = (0) (4),inner sep=6pt,transform shape=false] {};
        \end{pgfonlayer}
      \end{tikzpicture}
    \end{center}
    \caption{Labelling a signature.}
	\label{fig:labelling}
  \end{minipage}
  \hspace{-0.8cm}%
  \begin{minipage}[t]{.65\textwidth}
	\begin{center}
     \begin{tikzpicture}[scale=0.55,transform shape,semithick]
        \draw (0,0) node [] (0) {};
        \draw (0,3) node [] (4) {};
        \draw (4,0) node [] (g0) {};
        \draw (4,3) node [] (g4) {};

        \draw (-1.7,0) -- (-0.52,0);
        \draw (-1.7,2) -- (-0.52,2);
        \draw (-1.7,3) -- (-0.52,3);
        \draw (-1.3,0.75) node [] {\Huge $\vdots$};
        \draw (-1.3,1.25) node [] {\Huge $\vdots$};
        \draw (3.48,0) -- (0.52,0);
        \draw (3.48,2) -- (0.52,2);
        \draw (3.48,3) -- (0.52,3);
        \draw (2,0.75) node [] {\Huge $\vdots$};
        \draw (2,1.25) node [] {\Huge $\vdots$};
        \draw (5.7,0) -- (4.52,0);
        \draw (5.7,2) -- (4.52,2);
        \draw (5.7,3) -- (4.52,3);
        \draw (5.3,0.75) node [] {\Huge $\vdots$};
        \draw (5.3,1.25) node [] {\Huge $\vdots$};

        \draw (-2.2,3) node [] {\Large $x_1$};
        \draw (-2.2,2) node [] {\Large $x_2$};
        \draw (-2.2,0) node [] {\Large $x_n$};
        
        \draw (1,3.3) node [] {\Large $x_{2n}$};
        \draw (1.2,2.3) node [] {\Large $x_{2n-1}$};
        \draw (1.1,0.3) node [] {\Large $x_{n+1}$};

        \draw (3.2,3.3) node [] {\Large $y_{1}$};
        \draw (3.2,2.3) node [] {\Large $y_{2}$};
        \draw (3.2,0.3) node [] {\Large $y_{n}$};

        \draw (6.1,3) node [] {\Large $y_{2n}$};
        \draw (6.3,2) node [] {\Large $y_{2n-1}$};
        \draw (6.2,0) node [] {\Large $y_{n+1}$};

        \draw (0,1.5) node [] {\LARGE $f$};
        \draw (4,1.5) node [] {\LARGE $g$};
        \draw (2,4.5) node [] {\LARGE $h$};

        \begin{pgfonlayer}{background}
          \node[draw=black,thick,rounded corners,fit = (0) (4),inner sep=6pt,transform shape=false] {};
          \node[draw=black,thick,rounded corners,fit = (g0) (g4),inner sep=6pt,transform shape=false] {};
          \node[inner sep=12pt,transform shape=false,draw=\borderColor,thick,rounded corners,fit = (0) (4) (g0) (g4) ] {};
        \end{pgfonlayer}
      \end{tikzpicture}
	\end{center}
    \caption{Sequential composition of signatures.}
	\label{fig:seq}
  \end{minipage}
\end{figure}

The first $n$ inputs of $f$ are the $n$ inputs of $M_f$ and index a row while the last $n$ inputs of $f$ are the $n$ outputs of $M_f$ and index a column.
We assume that the variables are labelled counter-clockwise (see Fig.~\ref{fig:labelling}).
This requires us to reverse the order of the bits corresponding to the column index.
This is for convenience so that the signature matrix of the signature resulting from the direct linking of two arity $2n$ signatures
along a particular sequence of $n$ inputs is the matrix product of the signature matrices of the two signatures.
More precisely, for two arity $2n$ signatures $f$ and $g$ on inputs $x_1, \dotsc, x_{2n}$ and $y_1, \dotsc, y_{2n}$ respectively, 
the identification of $x_{2n+1-k} = y_k$ for $1 \le k \le n$ gives another arity $2n$ signature $h$ (see Fig.~\ref{fig:seq}) such that
\begin{align*}
  h(x_1, \dotsc, x_n, y_{n+1}, \dotsc, y_{2n})
  = \sum_{x_{n+1} = y_n, \dotsc, x_{2n} = y_1 \in \{0,1\}} f(x_1, \dotsc, x_{2n}) g(y_1, \dotsc, y_{2n}),
\end{align*}
and $M_h = M_f M_g$.

Using Definition~\ref{def:sig_matrix},
we define the set of signature matrices of $\mathscr{A}$ with $n$ inputs
 and $n$ outputs as
\begin{align*}
 \mathcal{A}_n = \{M_f \mid f \in \mathscr{A} \text{ and } \arity(f) = 2n\},
\end{align*}
and the subset of these signature matrices that are also unitary as
\begin{align*}
 \mathcal{UA}_n = U(2^n) \cap \mathcal{A}_n,
\end{align*}
where $U(k)$ is the group of $k$-by-$k$ unitary matrices.
Recall that a matrix $M$ is unitary if $M M^* = I$, where $M^*$ is the conjugate transpose of $M$.
For a matrix to be unitary, it must be at least nonsingular, so we also define
\begin{align*}
 \mathcal{GA}_n = GL_{2^n}(\mathbb{C}) \cap \mathcal{A}_n,
\end{align*}
where $GL_{k}(\mathbb{C})$ is the group of $k$-by-$k$ nonsingular matrices with complex entries.

In Lemma~\ref{lem:affine:unitary}, we prove that for every $M_f \in \mathcal{GA}_n$,
there exists a scalar $\lambda$ (the ``correct'' choice of scaling) such that $M_{\lambda f} \in \mathcal{UA}_n$.
The reverse direction $\mathcal{UA}_n \subseteq \mathcal{GA}_n$ is obvious.

\subsection{Clifford Gates} \label{sec:clifford}

The \emph{Pauli matrices} on a single qubit are
\[
 I = \begin{bmatrix} 1 &  0 \\ 0 &  1 \end{bmatrix},
 \qquad
 X = \begin{bmatrix} 0 &  1 \\ 1 &  0 \end{bmatrix},
 \qquad
 Y = \begin{bmatrix} 0 & -i \\ i &  0 \end{bmatrix},
 \qquad
 Z = \begin{bmatrix} 1 &  0 \\ 0 & -1 \end{bmatrix}.
\]
These matrices satisify
\[X^2 = Y^2 = Z^2 = I\]
and
\begin{alignat*}{3}
 X Y &=  i Z &       Y Z &=  i X &       Z X &=  i Y \\
 Y X &= -i Z &\qquad Z Y &= -i X &\qquad X Z &= -i Y.
\end{alignat*}
In particular, every two Pauli matrices either commute or anticommute.
There is an isomorphism with the quaternions if we send $(I,iZ,iY,iX)$ to $(1,i,j,k)$.

On $n$ qubits, the set of Pauli matrices is $P_n = \{\sigma_1 \otimes \dotsb \otimes \sigma_n \mid \sigma_k \in \{I, X, Y, Z\}\}$.
The group generated by these Pauli matrices forms the \emph{Pauli group} $\mathcal{P}_n = \langle P_n \rangle$.
This Pauli group can also be defined as simply the Pauli matrices together with a multiplicative factor of $\pm 1$ or $\pm i$.
A $2^n$-by-$2^n$ matrix $P$ is a matrix in $\mathcal{P}_n$ if and only if the entry $P(\mathbf{x},\mathbf{y})$ can be expressed as
\begin{align}\label{Pauli-entry-expression}
 P(\mathbf{x}, \mathbf{y}) = i^c \cdot \xi(\mathbf{x}, \mathbf{y}) \cdot(-1)^{\transpose{\mathbf{r}} \mathbf{x}}
\end{align}
where $c$ is a constant, $\mathbf{r} \in \mathbb{F}_2^n$ is a vector,
and $\xi$ is an indicator function of the linear system
determined by a sequence of 0-1 values $\epsilon_i \in \mathbb{F}_2$
such that for each $i \in [n]$, $x_i - y_i = \epsilon_i$ over $\mathbb{F}_2$.
We can rewrite the constraint induced by $\xi$ as
\begin{align}
 \mathbf{x} - \mathbf{y} = \mathbf{e}, \label{eqn:pauli:support}
\end{align}
where $\mathbf{e} \in \mathbb{F}_2^n$ is a constant vector determined by $P$.

\begin{definition}
 The \emph{Clifford group} $\mathcal{C}_n$ on $n$ qubits is
 \[\mathcal{C}_n = \{U \in U(2^n) \mid \sigma \in \mathcal{P}_n \implies U \sigma U^* \in \mathcal{P}_n\} / U(1).\]
\end{definition}

Here $U(1)$ is the group of complex numbers of norm one, $\{z \in \mathcal{C} \mid |z| = 1\}$.
It is considered as a subgroup by identifying $z$ with $z I_{2^n}$.
Alternatively, we can also consider that $\mathcal{C}_n$ is defined as
$\{U \in U(2^n) \mid [|\det(U)| = 1] \text{ and } [\sigma \in \mathcal{P}_n \implies U \sigma U^* \in \mathcal{P}_n] \}$.
In words, the Clifford group on $n$ qubits is, up to a norm~1 scalar,
the set of $2^n$-by-$2^n$ unitary matrices that stablizes the Pauli group $\mathcal{P}_n$.
Thus, quantum circuits composed by gates from the Clifford group are known as stabilizer circuits.

Having defined both (unitary) affine signatures and Clifford gates,
we are now ready to state our main theorem.

\begin{theorem} \label{thm:equiv}
  The set of non-singular affine signatures is exactly the set of Clifford gates,
  namely, $\mathcal{GA}_n /U(1) = \mathcal{C}_n$.
\end{theorem}

In particular, Theorem \ref{thm:equiv} implies that all non-singular affine signatures are unitary (see Corollary~\ref{cor:nonsingluar:unitary}).

\section{Closure of Affine Signatures}

We will show a closure property of affine signatures (see Lemma \ref{lem:affine:simple-operations}).
This property, together with the fact that the Clifford gates are generated by three elements (all of which are affine),
implies one direction of Theorem \ref{thm:equiv}.

Affine signatures $\mathscr{A}$ are closed under four basic types of operations.

\begin{lemma} \label{lem:affine:simple-operations}
 If $f(x_1, \dotsc, x_n), g(y_1, \dotsc, y_m) \in \mathscr{A}$,
 then so are
 \begin{enumerate}
   \item \label{case:tensor} $(f \otimes g)(x_1, \dotsc, x_n,  y_1, \dotsc, y_m) = f(x_1, \dotsc, x_n) g(y_1, \dotsc, y_m)$,
   \item \label{case:permute} $f(x_{\sigma(1)}, \dotsc, x_{\sigma(n)})$ for any permutation $\sigma \in S_n$,
   \item \label{case:identify} $f^{x_j = x_\ell} = f(x_1, \dotsc, x_{j - 1},
 x_{\ell}, x_{j + 1}, \dotsc, x_n)$, where we set the variable $x_j$ to be
equal to $x_{\ell}$,  and
   \item \label{case:marginal} $f^{x_j = *} = 
\sum_{x_j = 0,1} f(x_1, \dotsc, x_j, \dotsc, x_n)$.
 \end{enumerate}
\end{lemma}

\begin{proof}
 In the following, let $f = \lambda \chi_{A \mathbf{x}} i^{Q(\mathbf{x})}$ and $g = \mu \chi_{B \mathbf{y}} i^{P(\mathbf{y})}$,
 where $\mathbf{x} = (x_1, \dotsc, x_n, 1)$, $\mathbf{y} = (y_1, \dotsc, y_n, 1)$,
 and the functions $Q$ and $P$ are homogeneous quadratic polynomials over $\mathbb{Z}$ with even coefficients on cross terms.
 
 \begin{enumerate}
  \item The signature $f \tensor g$ is also affine since
  \begin{align*}
   f(\mathbf{x}) \tensor g(\mathbf{y})
   = \lambda \chi_{A \mathbf{x}} i^{Q(\mathbf{x})} \mu \chi_{B \mathbf{y}} i^{P(\mathbf{y})}
   = \lambda \mu \chi_{C \mathbf{z}} i^{Q(\mathbf{x}) + P(\mathbf{y})},
  \end{align*}
  where $C = \begin{bmatrix} A \\ B \end{bmatrix}$ and $\mathbf{z} = \begin{bmatrix} \mathbf{x} \\ \mathbf{y} \end{bmatrix}$.
  
  \item This is obvious.
 
  \item We simply add the affine equation $x_j = x_\ell$ to the linear system $A \mathbf{x} = 0$ of constraints that defines the affine support.
  Therefore, the resulting signature is also affine.
  Furthermore, we can decrease the arity by replacing all occurrences 
of  $x_j$ with $x_\ell$.
 
  \item If there is a nontrivial equation in $A \mathbf{x} = 0$ involving $x_j$,
  then for any setting of the other variables,
  there is at most one value for $x_j$ such that the resulting $\mathbf{x}$ vector is a solution to $A \mathbf{x} = 0$.
  Thus in place of the sum over every possible value for $x_j$,
  we use this nontrivial equation to replace $x_j$ in both $A \mathbf{x} = 0$ and $Q$.
As noted earlier, for $Q(\mathbf{x}) = \transpose{ \mathbf{x}} S \mathbf{x}$
for some symmetric $S$, the substitution of $x_j$ by a mod 2 expression
is valid.
  Then the resulting function is also affine.
 
  Otherwise, there is no nontrivial equation involving $x_j$.
  In this case, the fact that $f^{x_j=*}$ is also affine essentially corresponds to why the affine signatures are tractable in the first place.
  See the proof of \cite[Theorem 4.1]{CLX14},
  which uses the original definition of the affine signatures where the exponent of $i$ is a sum of linear indicator functions. \qedhere
 \end{enumerate}
\end{proof}



On the other hand, the Clifford group is generated by (essentially) three elements.
Let
\begin{align*}
  H    := \frac{1}{\sqrt{2}} \begin{bmatrix} 1 & 1 \\ 1 & -1 \end{bmatrix};
  \qquad
  P    :=                    \begin{bmatrix} 1 & 0 \\ 0 &  i \end{bmatrix};
  \qquad
  CNOT :=  \begin{bmatrix} 1 & 0 & 0 & 0 \\ 0 & 1 & 0 & 0 \\ 0 & 0 & 0 & 1 \\ 0 & 0 & 1 & 0 \end{bmatrix}.  
\end{align*}
Furthermore, let $H_j$, $P_j$, and $CNOT_{jk}$ be the gate that acts on the subscript qubits (either on $j$, or on $j$ and $k$).
In other words, $H_j = I\otimes \dots \otimes H \otimes \dots I $ where $I$ is the $2$-by-$2$ identity matrix and $H$ is at the $j$th coordinate.
$P_j$ can be expressed similarly, and $CNOT_{jk}$ can be viewed as the gate $CNOT\otimes I\otimes\dots\otimes I$ after
 permuting qubits $1$ to $j$ and $2$ to $k$.

\begin{theorem}[Page~13 in~\cite{Got98}] \label{thm:clifford:generators}
 $\mathcal{C}_n = \langle H_j, P_j, CNOT_{jk} \rangle / U(1)$.
\end{theorem}

Firstly, these three basic matrices are signature matrices of affine signatures:
\begin{align*}
 H(x_1, x_2) &= \frac{1}{\sqrt{2}} i^{2 x_1 x_2};\\
 P(x_1, x_2) &= \chi_{x_1 = x_2} i^{x_1^2};\\
 CNOT(x_1, x_2, x_3, x_4) &= \chi_{x_1 = x_4 = x_2 + x_3}.
\end{align*}
By case \ref{case:tensor} of Lemma~\ref{lem:affine:simple-operations}, $H_j$ and $P_j$ are in $\mathcal{GA}_n$.
By case \ref{case:tensor} and \ref{case:permute} of Lemma~\ref{lem:affine:simple-operations}, $CNOT_{jk}$ are also in $\mathcal{GA}_n$.

Furthermore, a matrix multiplication is a combination of operations \ref{case:identify} and \ref{case:marginal} in Lemma~\ref{lem:affine:simple-operations}.
Thus Lemma~\ref{lem:affine:simple-operations} implies that $\langle H_j, P_j, CNOT_{jk} \rangle\subseteq \mathcal{GA}_n$.
By Theorem~\ref{thm:clifford:generators}, we have the following lemma, which is similar to a result in~\cite{DD03}.

\begin{lemma}
  $\mathcal{C}_n \subseteq \mathcal{GA}_n /U(1)$.
  \label{lem:clifford:inclusion}
\end{lemma}

\section{Equivalence of Clifford Gates and Unitary Affine Signatures}

At last, we show that the reverse direction of Theorem \ref{thm:equiv} also holds.

For a signature matrix $M_f$, assume the rows are indexed by $\mathbf{x}\in\{0,1\}^n$ and the columns by $\mathbf{y}\in\{0,1\}^n$.
Suppose $M_f$ is nonsingular.
Since $f\in\mathscr{A}$, the support of $f$ is determined by a linear system.
If there is any non-trivial equation involving the $x_i$'s only, there will be a row of entire $0$'s. 
This is a contradiction to the assumption of $M_f$ being nonsingular.
Suppose the support of $f$ has dimension $m$. One can choose
a set of $m$ free variables among $\{x_1, \ldots, x_n,
y_1, \ldots, y_n\}$ such that every variable can be 
expressed as an affine linear combination over $\mathbb{F}_2$
of these $m$ variables.
Among such choices of a set of free variables,  pick one set $S$
with a maximum number of variables from $\{x_1, \ldots, x_n\}$.
We claim that $\{x_1, \ldots, x_n\} \subseteq S$.
Suppose otherwise. There is some  $j \in [n]$ such that $x_j \not \in S$. Then
the expression of $x_j$ in terms of the variables in $S$
must contain some $y_k$ with a non-zero coefficient, or else
there would be a non-trivial equation involving the $x_i$'s only.
This equation can be used to exchange $y_k$ by $x_j$ and get a new 
set of free variables $S'$,
contradicting the maximality of $S$.
Hence $\{x_1, \ldots, x_n\} \subseteq S$.

By renaming the variables $y_1, \ldots, y_n$ if necessary, 
we may assume that for some $0 \le r \le n$,
\[S = \{x_1, \ldots, x_n,  y_1, \ldots, y_r\}.\]
Then we have the following system in $\mathbb{Z}_2$, if $r<n$,
\begin{align*}
  \begin{pmatrix}
	y_{r+1}\\
	y_{r+2}\\
	\vdots\\
	y_{n}
  \end{pmatrix}
  = A
  \begin{pmatrix}
	x_1\\
	x_2\\
	\vdots\\
	x_n
  \end{pmatrix}  
  + B
  \begin{pmatrix}
	y_{1}\\
	y_{2}\\
	\vdots\\
	y_{r}
  \end{pmatrix}
  + \mathbf{b},
\end{align*}
where $A$ is an $(n-r)$-by-$n$ matrix, $B$ is an $(n-r)$-by-$r$ matrix, and $\mathbf{b}$ is a vector of length $n-r$.
(If $r=0$ then the $B$ term disappears.  
If $r=n$ then this linear system is empty.)
We will denote the variables 
$\mathbf{y}'  = \transpose{(y_1, \ldots, y_{r})}$ and
$\mathbf{y}'' = \transpose{(y_{r+1}, \ldots, y_{n})}$.
Then
\begin{align}
 \mathbf{y}''= A \mathbf{x} + B \mathbf{y}' + \mathbf{b}. 
 \label{eqn:affine:support}
\end{align}

Each entry of $f$ is of the following form:
\begin{align}\label{eqn:f-expression}
  f(\mathbf{x},\mathbf{y})=\lambda \chi(\mathbf{x},\mathbf{y})  i^{Q(\mathbf{x},\mathbf{y})}
\end{align}
where $Q$ is a homogeneous quadratic form in $\mathbf{x}$ and $\mathbf{y}$ 
with all cross terms having even coefficients,
and $\chi(\mathbf{x},\mathbf{y})$ is the 0-1 indicator function of
the support given in (\ref{eqn:affine:support}).
The exponent is evaluated mod $4$.
As discussed in Section~\ref{subsec:affine_def},
we can use equations~\eqref{eqn:affine:support} to substitute $\mathbf{y}''$ by $\mathbf{x}$ and $\mathbf{y}'$.
The substitution in (\ref{eqn:f-expression}) by (\ref{eqn:affine:support})
is valid because every cross term has an even coefficient
and for any integer $x$, we have $x \equiv 0 ~\mbox{or}~1 \pmod 2$
iff $x^2 \equiv 0 ~\mbox{or}~1 \pmod 4$.
Then we have
\begin{align}
  f(\mathbf{x}, \mathbf{y}', \mathbf{y}'')
  = \mu \chi(\mathbf{x},\mathbf{y}',\mathbf{y}'') i^{\transpose{\mathbf{x}} C_1 \mathbf{x} + 
  \transpose{\mathbf{y'}} C_2 \mathbf{y}' 
+ 2 \transpose{\mathbf{y'}} C_3 \mathbf{x}} 
  \label{eqn:affine:expression}
\end{align}
where $C_1$ is an $n$-by-$n$ symmetric matrix, $C_2$ is an $r$-by-$r$ 
symmetric matrix
and $C_3$ is an $r$-by-$n$ matrix, all with integer entries. 
Moreover, $\mu$ equals $\lambda$ multiplied by a power of $i$.
In particular, $|\mu| = |\lambda|$.
Here we use $\chi(\mathbf{x},\mathbf{y}',\mathbf{y}'')$ to mean $\chi(\mathbf{x},\mathbf{y})$ where $\mathbf{y}$ is the concatenation of $\mathbf{y}'$ and $\mathbf{y}''$.
If we attach the matrix $C_3$ with $A$, we can form a square matrix $C_f$:
\begin{align*}
  C_f=
  \begin{pmatrix}
	A\\
	C_3
  \end{pmatrix}.
\end{align*}

\begin{lemma} \label{lem:affine:unitary}
  For a signature $f \in \mathscr{A}$ given in \eqref{eqn:f-expression},
  if $M_f$ is nonsingular, with $|\lambda| = 2^{-r/2}$, 
  where $r$ is the dimension of the affine support of $f$,
  then $M_f$ is unitary and $C_f$ is nonsingular.
\end{lemma}
\begin{proof}
  Given a nonsingular $M_f$,
  let $A$, $B$, $\mathbf{b}$, $\mathbf{x}$, $\mathbf{y}'$ and $\mathbf{y}''$ be defined as above.

  Consider the matrix $N=M_f M_f^*$. 
  We will show that $N$ is the identity matrix $I_{2^n}$.
  For the entry $N_{\mathbf{u},\mathbf{v}}$ indexed by row $\mathbf{u}$ and column $\mathbf{v}$, it has the following expression
  \begin{align*}
	N_{\mathbf{u},\mathbf{v}}= 2^{-r}
	\sum_{\substack{\mathbf{y}'\in\{0,1\}^{r}\\\mathbf{y}''\in\{0,1\}^{n-r}}}
	&
	\chi(\mathbf{u},\mathbf{y}',\mathbf{y}'') i^{\transpose{\mathbf{u}} C_1 \mathbf{u} + 
	\transpose{\mathbf{y'}} C_2 \mathbf{y}' + 2 \transpose{\mathbf{y'}} C_3 \mathbf{u}}\\
	&\cdot\chi(\mathbf{v},\mathbf{y}',\mathbf{y}'') i^{-\transpose{\mathbf{v}} C_1 \mathbf{v} - 
	\transpose{\mathbf{y'}} C_2 \mathbf{y}' - 2 \transpose{\mathbf{y'}} C_3 \mathbf{v}}.
  \end{align*}
  Notice that if $\chi(\mathbf{u},\mathbf{y}',\mathbf{y}'')\chi(\mathbf{v},\mathbf{y}',\mathbf{y}'')\neq 0$, then the following equations hold
  \begin{align}	\label{eqn:support1}
	\mathbf{y}'' = A \mathbf{u} + B \mathbf{y}' + \mathbf{b},
  \end{align}
  and
  \begin{align}	\label{eqn:support2}
	\mathbf{y}'' = A \mathbf{v} + B \mathbf{y}' + \mathbf{b}.
  \end{align}
  In this case, these two systems must be consistent and we have
  \begin{align}	\label{eqn:support:consistent}
	A(\mathbf{u}-\mathbf{v})=0.
  \end{align}
  When this is indeed the case, any solution ($\mathbf{y}',\mathbf{y}''$) to \eqref{eqn:support1} is also a solution to \eqref{eqn:support2}, and vice versa.
  Also notice that for $\mathbf{u}=\mathbf{v}$, condition \eqref{eqn:support:consistent} always holds.
  On the other hand, if \eqref{eqn:support:consistent} does not hold, then $N_{\mathbf{u},\mathbf{v}}=0$.
  We may assume that \eqref{eqn:support:consistent} holds.
  Furthermore, we may assume that $\mathbf{y}''$ is substituted by \eqref{eqn:support1}, 
  which is the same as \eqref{eqn:support2} under \eqref{eqn:support:consistent}, and drop the indicator function $\chi$.
  Since $\lambda$ is chosen so that $\mu \overline{\mu} = \lambda\overline{\lambda} = 1/2^r$ (recall \eqref{eqn:affine:expression}), we have that
  \begin{align}
	2^r N_{\mathbf{u},\mathbf{v}}
	&= \sum_{\mathbf{y}'\in\{0,1\}^{r}}
	i^{\transpose{\mathbf{u}} C_1 \mathbf{u} + \transpose{\mathbf{y'}} C_2 \mathbf{y}' + 2 \transpose{\mathbf{y'}} C_3 \mathbf{u}
	-\transpose{\mathbf{v}} C_1 \mathbf{v} - \transpose{\mathbf{y'}} C_2 \mathbf{y}' - 2 \transpose{\mathbf{y'}} C_3 \mathbf{v}} \notag\\
	&= \sum_{\mathbf{y}'\in\{0,1\}^{r}}
	i^{(\transpose{\mathbf{u}} C_1 \mathbf{u}-\transpose{\mathbf{v}} C_1 \mathbf{v}) + 2 \transpose{\mathbf{y'}} C_3 (\mathbf{u} - \mathbf{v}) } \notag\\
	&= i^{(\transpose{\mathbf{u}} C_1 \mathbf{u}-\transpose{\mathbf{v}} C_1 \mathbf{v})}
	\sum_{\mathbf{y}'\in\{0,1\}^{r}} (-1)^{\transpose{\mathbf{y'}} C_3 (\mathbf{u} - \mathbf{v})} \label{eqn:entry:exp}
  \end{align}
  If $\mathbf{u}=\mathbf{v}$, then \eqref{eqn:entry:exp} simplifies into 
  \begin{align*}
	2^r N_{\mathbf{u},\mathbf{u}}
 	&=\sum_{\mathbf{y'}\in\{0,1\}^{r}} 1 =2^r.
  \end{align*}
  Therefore the diagonal entries in $N$ are all $1$.

  If $r=0$, then the matrix $C_f=A$ is a square matrix. 
  It must have full rank over $\mathbb{Z}_2$.
  Otherwise, there is a non-zero vector $\mathbf{w} \in \mathbb{Z}_2^n$ such that $\transpose{\mathbf{w}} A =0$ over $\mathbb{Z}_2^n$.
  This would give a non-trivial affine linear equation on $\mathbf{y}$ by \eqref{eqn:affine:support}. 
  Then $M_f$ would have an all 0 column, a contradiction.
  Therefore the condition \eqref{eqn:support:consistent} cannot hold for $\mathbf{u}\neq\mathbf{v}$.
  Hence all off-diagonal entries in $N$ are $0$.
  This implies that $N$ is the identity matrix, and $M_f$ is unitary.

  Otherwise $r>0$. 
  For $\mathbf{u}\neq\mathbf{v}$, by \eqref{eqn:entry:exp}, we have that
  \begin{align*}
	2^r N_{\mathbf{u},\mathbf{v}}
	&=i^{(\transpose{\mathbf{u}} C_1 \mathbf{u}-\transpose{\mathbf{v}} C_1 \mathbf{v})}
	\sum_{\mathbf{y}'\in\{0,1\}^{r}} (-1)^{\transpose{\mathbf{y'}} C_3 (\mathbf{u} - \mathbf{v})}\\
	&=i^{(\transpose{\mathbf{u}} C_1 \mathbf{u}-\transpose{\mathbf{v}} C_1 \mathbf{v})}
	\sum_{\mathbf{y}'\in\{0,1\}^{r}} \prod_{k=1}^r (-1)^{y_k \mathbf{c}_k(\mathbf{u} - \mathbf{v})}\\
	&=i^{(\transpose{\mathbf{u}} C_1 \mathbf{u}-\transpose{\mathbf{v}} C_1 \mathbf{v})}
	\prod_{k=1}^r \sum_{y_k\in\{0,1\}}  ((-1)^{\mathbf{c}_k{(\mathbf{u} - \mathbf{v})}})^{y_k},
  \end{align*}
  where $\mathbf{c}_k$ is the $k$th row of $C_3$.

  If $C_3 (\mathbf{u} - \mathbf{v})\neq\mathbf{0}$, it is easy to see that $N_{\mathbf{u},\mathbf{v}}=0$.
  Otherwise, $C_3 (\mathbf{u} - \mathbf{v})=\mathbf{0}$ and we argue that $M_f$ is singular.
  In fact, the rows $\mathbf{u}$ and $\mathbf{v}$ are linearly dependent.
  First, since condition \eqref{eqn:support:consistent} holds, 
  an entry indexed by $(\mathbf{y}',\mathbf{y}'')$ in row $\mathbf{u}$ is nonzero if and only if 
  the corresponding entry in  row $\mathbf{v}$ is.
  For a solution $(\mathbf{y}',\mathbf{y}'')$ to the system \eqref{eqn:support1} and equivalently \eqref{eqn:support2}, 
  we have that
  \begin{align*}
	\frac{f(\mathbf{u},\mathbf{y}',\mathbf{y}'')}{f(\mathbf{v},\mathbf{y}',\mathbf{y}'')}
	&=i^{\transpose{\mathbf{u}} C_1 \mathbf{u} + \transpose{\mathbf{y'}} C_2 \mathbf{y}' + 2 \transpose{\mathbf{y'}} C_3 \mathbf{u}
	-\transpose{\mathbf{v}} C_1 \mathbf{v} - \transpose{\mathbf{y'}} C_2 \mathbf{y}' - 2 \transpose{\mathbf{y'}} C_3 \mathbf{v}}\\
	&=i^{(\transpose{\mathbf{u}} C_1 \mathbf{u}-\transpose{\mathbf{v}} C_1 \mathbf{v})}
	  (-1)^{\mathbf{y'}C_3{(\mathbf{u} - \mathbf{v})}}\\
	&=i^{(\transpose{\mathbf{u}} C_1 \mathbf{u}-\transpose{\mathbf{v}} C_1 \mathbf{v})},
  \end{align*}
  where in the last equality we used $C_3 (\mathbf{u} - \mathbf{v})=\mathbf{0}$.
  This ratio is independent of $\mathbf{y}$.
  Therefore the two rows indexed by $\mathbf{u}$ and $\mathbf{v}$ are linearly dependent and the matrix $M_f$ is singular. Contradiction.

  To conclude, $N_{\mathbf{u},\mathbf{v}}=1$ if $\mathbf{u} = \mathbf{v}$,
  and $N_{\mathbf{u},\mathbf{v}}=0$ otherwise.
  This is to say that $N$ is the identity matrix and $M_f$ is unitary.
  Moreover, $\mathbf{u}=\mathbf{v}$ is the only solution to the system
  \begin{align*}
  \begin{pmatrix}
	A\\
	C_3
  \end{pmatrix}(\mathbf{u}-\mathbf{v})
  =C_f(\mathbf{u}-\mathbf{v})=\mathbf{0}.
  \end{align*}
  So $C_f$ is nonsingular.
\end{proof}

Since unitary matrices are non-singular, it follows that $\mathcal{UA}_n \subseteq \mathcal{GA}_n$.
Lemma \ref{lem:affine:unitary} leads to the following corollary.
\begin{corollary}  \label{cor:nonsingluar:unitary}
  $\mathcal{GA}_n/U(1)=\mathcal{UA}_n /U(1)$.
\end{corollary}

Now we prove that nonsingular signature matrices of affine signatures are gates in the Clifford group.
\begin{lemma} \label{lem:affine:inclusion}
 $\mathcal{GA}_n/U(1) \subseteq\mathcal{C}_n$.
\end{lemma}

\begin{proof}
  For a given matrix $M_f\in\mathcal{GA}_n$ and $P\in\mathcal{P}_n$, we want to show that $N=M_f P M_f^*\in\mathcal{P}_n$.
Clearly we may assume $c=0$ in the expression \eqref{Pauli-entry-expression}
for entries of $P$.
  The entry $N_{\mathbf{u},\mathbf{v}}$ at row $\mathbf{u}$ and column $\mathbf{v}$ can be expressed as
  \begin{align*}
     N_{\mathbf{u},\mathbf{v}}  
  = 
	\sum_{\mathclap{\substack{\mathbf{x}', \mathbf{y}' \in \{0,1\}^{r}\\ \mathbf{x}'', \mathbf{y}'' \in \{0,1\}^{n-r}}}}\;
	 & 2^{-r}
	\xi(\mathbf{x},\mathbf{y})
	\chi(\mathbf{u},\mathbf{x'},\mathbf{x}'') 
	(-1)^{\transpose{\mathbf{r}} \mathbf{x}}
    \cdot i^{\transpose{\mathbf{u}} C_1 \mathbf{u} + 
	\transpose{\mathbf{x'}} C_2 \mathbf{x}' + 2 \transpose{\mathbf{x'}} C_3 \mathbf{u}}\\
       &\cdot\chi(\mathbf{v},\mathbf{y}',\mathbf{y}'')(-1)^{\transpose{\mathbf{r}} \mathbf{y}}
	\cdot i^{-\transpose{\mathbf{v}} C_1 \mathbf{v} - 
	\transpose{\mathbf{y'}} C_2 \mathbf{y'} - 2 \transpose{\mathbf{y'}} C_3 \mathbf{v}}, 
  \end{align*}
  where $C_1$, $C_2$ and $C_3$
  come from the general expression \eqref{eqn:affine:expression} of the entry in $M_f$.

  If $\xi(\mathbf{x},\mathbf{y})\chi(\mathbf{u},\mathbf{x'},\mathbf{x}'')\chi(\mathbf{v},\mathbf{y}',\mathbf{y}'')\neq 0$,
  then the following equations must hold
  \begin{align*}
	\mathbf{x}'' &= A \mathbf{u} + B \mathbf{x}' + \mathbf{b},\\
	\mathbf{y}'' &= A \mathbf{v} + B \mathbf{y}' + \mathbf{b},\\
	\mathbf{x}' - \mathbf{y}' &= \mathbf{e}',\\
	\mathbf{x}'' - \mathbf{y}'' &= \mathbf{e}'',
  \end{align*}
  where we write $\mathbf{e} = \begin{pmatrix} \mathbf{e}' \\ \mathbf{e}''\end{pmatrix}$ (recall \eqref{eqn:pauli:support}).
  The first two constraints come from the affine support constraint \eqref{eqn:affine:support}
  whereas the last two come from the Pauli matrix constraint \eqref{eqn:pauli:support}.
  Therefore we have 
  \begin{align}	\label{eqn:pauli:affine}
	A(\mathbf{u}-\mathbf{v}) = \mathbf{e}'' - B \mathbf{e}'.
  \end{align}
  When this is the case and $\lambda$ is normalized so that $\lambda\overline{\lambda}=1/2^r$, we have that
  \begin{align*}
	2^r N_{\mathbf{u},\mathbf{v}} & =
	\sum_{\substack{\mathbf{x}'\in\{0,1\}^{r}}}
	(-1)^{\transpose{\mathbf{r}} \mathbf{x}+\transpose{\mathbf{r}} (\mathbf{x+e})}\cdot
    i^{\transpose{\mathbf{u}} C_1 \mathbf{u} + 
	\transpose{\mathbf{x'}} C_2 \mathbf{x}' + 2 \transpose{\mathbf{x'}} C_3 \mathbf{u}
	-\transpose{\mathbf{v}} C_1 \mathbf{v} - 
	\transpose{(\mathbf{x'+e'})} C_2 (\mathbf{x'+e'}) - 2 \transpose{(\mathbf{x'+e'})} C_3 \mathbf{v}}\\
	&=(-1)^{\transpose{\mathbf{r}} \mathbf{e}}\cdot
	i^{\transpose{\mathbf{u}} C_1 \mathbf{u}-\transpose{\mathbf{v}} C_1 \mathbf{v}}
	\sum_{\substack{\mathbf{x}'\in\{0,1\}^{r}}}
    i^{\transpose{\mathbf{x'}} C_2 \mathbf{x}' + 2 \transpose{\mathbf{x'}} C_3 \mathbf{u}
	-\transpose{(\mathbf{x'+e'})} C_2 (\mathbf{x'+e'}) - 2 \transpose{(\mathbf{x'+e'})} C_3 \mathbf{v}}\\
	&=(-1)^{\transpose{\mathbf{r}} \mathbf{e}}\cdot
	i^{\transpose{\mathbf{u}} C_1 \mathbf{u}-\transpose{\mathbf{v}} C_1 \mathbf{v}}
	\sum_{\substack{\mathbf{x}'\in\{0,1\}^{r}}}
    i^{2 \transpose{\mathbf{x'}} C_3 \mathbf{u} - 2 \transpose{\mathbf{x'}} C_3 \mathbf{v}-2 \transpose{\mathbf{e'}} C_3 \mathbf{v}
	-\transpose{\mathbf{e'}} C_2 \mathbf{e'} - 2\transpose{\mathbf{x'}} C_2 \mathbf{e'}}\\
	&=(-1)^{\transpose{\mathbf{r}} \mathbf{e}-\transpose{\mathbf{e'}} C_3 \mathbf{v}}\cdot
	i^{\transpose{\mathbf{u}} C_1 \mathbf{u}-\transpose{\mathbf{v}} C_1 \mathbf{v}-\transpose{\mathbf{e'}} C_2 \mathbf{e'}}
	\sum_{\substack{\mathbf{x}'\in\{0,1\}^{r}}}
    (-1)^{ \transpose{\mathbf{x'}} C_3 \mathbf{u} -  \transpose{\mathbf{x'}} C_3 \mathbf{v}
	- \transpose{\mathbf{x'}} C_2 \mathbf{e'}}\\
	&=(-1)^{\transpose{\mathbf{r}} \mathbf{e}-\transpose{\mathbf{e'}} C_3 \mathbf{v}}\cdot
	i^{\transpose{\mathbf{u}} C_1 \mathbf{u}-\transpose{\mathbf{v}} C_1 \mathbf{v}-\transpose{\mathbf{e'}} C_2 \mathbf{e'}}
	\sum_{\substack{\mathbf{x}'\in\{0,1\}^{r}}}
    (-1)^{\transpose{\mathbf{x'}} (C_3 (\mathbf{u} - \mathbf{v}) - C_2 \mathbf{e'})}.
  \end{align*}
  Similar to the proof of Lemma~\ref{lem:affine:unitary}, the entry $N_{\mathbf{u},\mathbf{v}}\neq 0$ if and only if 
  the following equation holds:
  \begin{align*}
  \begin{pmatrix}
	A\\
	C_3
  \end{pmatrix}(\mathbf{u}-\mathbf{v})
  =C_f(\mathbf{u}-\mathbf{v})=
  \begin{pmatrix}
	\mathbf{e}'' - B \mathbf{e}'\\
	C_2 \mathbf{e'}.
  \end{pmatrix}
  \end{align*}
Suppose this equation holds.
  Since the matrix $C_f$ is nonsingular by Lemma~\ref{lem:affine:unitary},
 there exists a unique solution $\mathbf{t}$ that $\mathbf{u}-\mathbf{v}=\mathbf{t}$.
  It implies that
  \begin{align*}
	2^r N_{\mathbf{u},\mathbf{v}} & = (-1)^{\transpose{\mathbf{r}} \mathbf{e}-\transpose{\mathbf{e'}} C_3 (\mathbf{u-t})}\cdot
	i^{\transpose{\mathbf{u}} C_1 \mathbf{u}-\transpose{(\mathbf{u-t})} C_1 \mathbf{(u-t)}-\transpose{\mathbf{e'}} C_2 \mathbf{e'}}
	\cdot 2^r\\
	&=2^r(-1)^{\transpose{\mathbf{r}} \mathbf{e}-\transpose{\mathbf{e'}} C_3 (\mathbf{u-t})}\cdot
	i^{-\transpose{\mathbf{t}} C_1 \mathbf{t}+2\transpose{\mathbf{t}} C_1 \mathbf{u}-\transpose{\mathbf{e'}} C_2 \mathbf{e'}}\\
	&=2^r
	i^{-\transpose{\mathbf{t}} C_1 \mathbf{t}-\transpose{\mathbf{e'}} C_2 \mathbf{e'}}
	(-1)^{\transpose{\mathbf{r}} \mathbf{e}-\transpose{\mathbf{e'}} C_3 \mathbf{u}+\transpose{\mathbf{e'}} C_3 \mathbf{t}}\cdot
	(-1)^{\transpose{\mathbf{t}} C_1 \mathbf{u}}\\
	&=2^r
	i^{-\transpose{\mathbf{t}} C_1 \mathbf{t}-\transpose{\mathbf{e'}} C_2 \mathbf{e'}+2\transpose{\mathbf{e'}} C_3 \mathbf{t}+2\transpose{\mathbf{r}} \mathbf{e}}
	(-1)^{(\transpose{\mathbf{t}} C_1-\transpose{\mathbf{e'}} C_3)\mathbf{u}}.
  \end{align*}
  Therefore, we can express the entries of $N$ as
  \begin{align*}
	N_{\mathbf{u},\mathbf{v}}=i^{c'}\cdot\xi'(\mathbf{u},\mathbf{v})\cdot(-1)^{\transpose{\mathbf{r'}}\mathbf{u}},
  \end{align*}
  where 
  \begin{align*}
    c'=-\transpose{\mathbf{t}} C_1 \mathbf{t}-\transpose{\mathbf{e'}} C_2 \mathbf{e'}+2\transpose{\mathbf{e'}} C_3 \mathbf{t} + 2\transpose{\mathbf{r}} \mathbf{e},
  \end{align*}
  $\xi'(\mathbf{u},\mathbf{v})$ is the 0-1 indicator function of
  \begin{align*}
    \mathbf{u} - \mathbf{v} = \mathbf{t}, 
  \end{align*}
  and
  \begin{align*}
	\mathbf{r'}=\transpose{\mathbf{t}} C_1-\transpose{\mathbf{e'}} C_3.
  \end{align*}
  This matches the general expression of a Pauli matrix, so $N \in \mathcal{P}_n$.
\end{proof}

\begin{proof}[Proof of Theorem \ref{thm:equiv}]
  The theorem follows by combining Lemma~\ref{lem:clifford:inclusion} and Lemma~\ref{lem:affine:inclusion}.
\end{proof}

\bibliographystyle{plain}
\bibliography{bib}

\end{document}